\documentclass[11pt]{article}
\usepackage{graphicx}
\usepackage{amssymb,amsmath,amsthm}
\usepackage{fullpage}
\usepackage[left]{lineno}
\usepackage{adjustbox}
\usepackage[utf8]{inputenc}
\usepackage[T1]{fontenc}
\usepackage{siunitx}

\newcommand{\comment}[1]{} 

\newtheorem{theorem}{Theorem}

\newtheorem{lemma}[theorem]{Lemma}
\newtheorem{problem}{Problem}

\newcounter{quote}

\newcommand\blfootnote[1]{%
  \begingroup
  \renewcommand\thefootnote{}\footnote{#1}%
  \addtocounter{footnote}{-1}%
  \endgroup
}

\begin{document}

\title{On the Number of Order Types in Integer Grids of Small Size}
\author{Luis E. Caraballo\thanks{Departamento de Matemática Aplicada II, Universidad de Sevilla, Sevilla, Spain. \texttt{[lcaraballo|dbanez]@us.es}. } \footnote{Funded by Spanish Government under grant agreement FPU14/04705.}
\and
Jos\'e-Miguel D\'iaz-B\'a\~nez \footnotemark[1] \thanks{Partially supported by Project GALGO (Spanish Ministry of Economy and Competitiveness, MTM2016-
76272-R AEI/FEDER,UE)} 
\and
Ruy Fabila-Monroy\thanks{Departamento de Matem\'aticas, Cinvestav, CDMX, Mexico. \texttt{[rfabila|cmhidalgo]@math.cinvestav.mx}} \footnote{Partially supported by CONACYT (Mexico), grant 253261.} 
\and
Carlos Hidalgo-Toscano\footnotemark[4] \footnotemark[5]
\and
Jes\'us Lea\~nos\thanks{Unidad Académica de Matemáticas, Universidad Aut\'onoma de Zacatecas, Zacatecas, Mexico. \texttt{jleanos@matematicas.reduaz.mx}}
\and 
Amanda Montejano\thanks{UMDI-Juriquilla Facultad de Ciencias, Universidad Nacional Autónoma de México, Querétaro, Mexico. \texttt{amandamontejano@ciencias.unam.mx}}
}


\maketitle

\blfootnote{\begin{minipage}[l]{0.3\textwidth} \includegraphics[trim=10cm 6cm 10cm 5cm,clip,scale=0.15]{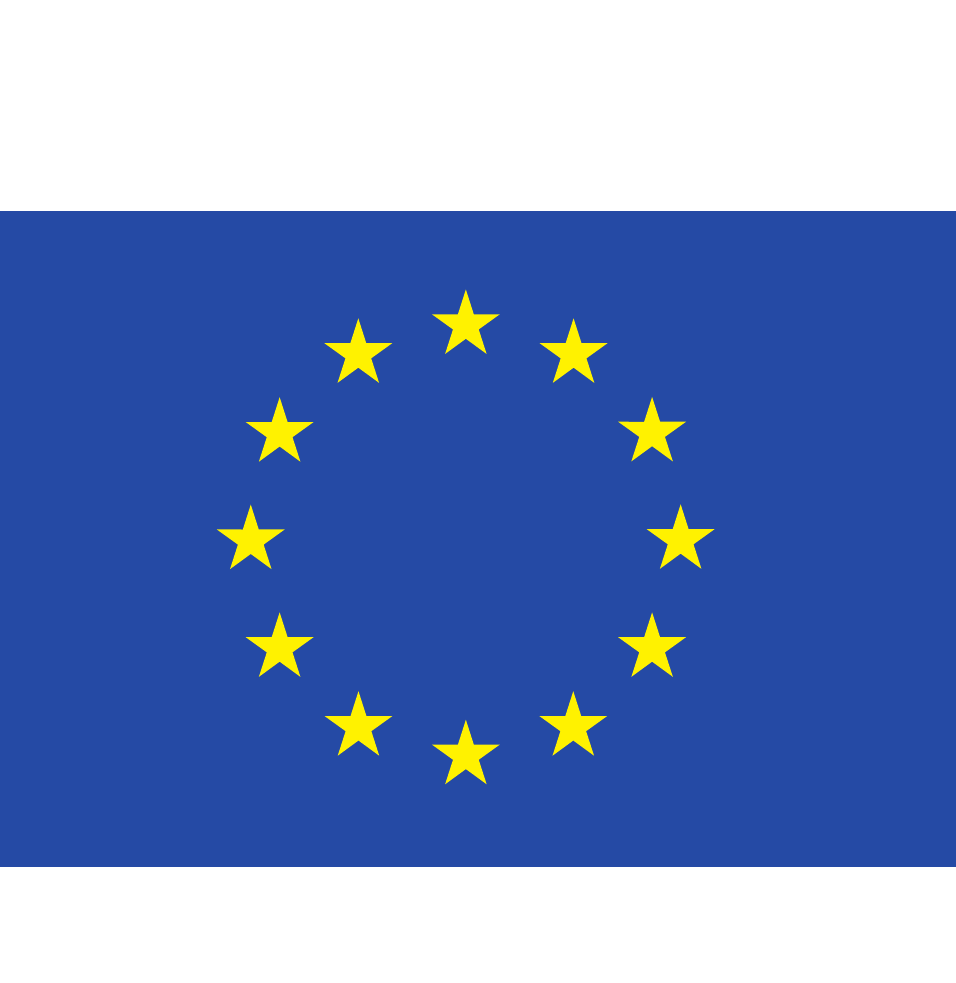} \end{minipage}  \hspace{-3.5cm} \begin{minipage}[l][1cm]{0.82\textwidth}
\vspace{.1cm}
This project has received funding from the European Union's Horizon 2020 research and innovation programme under the Marie Sk\l{}odowska-Curie grant agreement No 734922.
\end{minipage}}

\begin{abstract}
Let $\{p_1,\dots,p_n\}$  and $\{q_1,\dots,q_n\}$ be two sets of $n$ labeled points in general position in the plane. 
We say that these two point sets have the same order type if for every triple of indices $(i,j,k)$, $p_k$ is
above the directed line from $p_i$ to $p_j$ if and only if $q_k$ is
above the directed line from $q_i$ to $q_j$.
In this paper we give the first non-trivial lower bounds on the number
of different order types of $n$ points that can be realized in integer grids of polynomial
size.
\end{abstract}

\section{Introduction}

Let $A$ and $B$ be two arrays of $n$  distinct numbers. 
We say that $A$ and $B$ have the same order type 
if for every pair $i,j$ of different indices we have that  $A[i] < A[j]$ if and only if 
$B[i] < B[j]$. Goodman and Pollack~\cite{multidimSorting} introduced a higher dimensional analogue of this idea.
Let $S:=\{p_1,\dots,p_n\}$ be a set of $n$ labeled points in general position in the plane.
The relationship that $A[i] < A[j]$ is equivalent to $A[i]$ being the left of $A[j]$ in the real line. 
This left-right relationship can be generalized to point sets as follows. For a given triple
$(i,j,k)$ of distinct indices, $p_k$ may be above or below
the directed line from $p_i$ to $p_j$. Two sets of $n$ labeled points
in the plane have the same \emph{order type} if they have these same above-below relationships.\footnote{In the literature, it is more common to consider two points sets as having the same
order type if there is a bijection between them that preserves these above-below
relationships. In this paper we only consider \emph{labeled} order types; 
thus, a relabeling of $S$ usually produces a different order type.}
In dimension $d>2$ this 
is generalized by considering all $(d+1)$ tuples of points of a given labeled point set in $\mathbb{R}^d$.
The first $d$ points of the tuple define an oriented hyperplane and the relationship is whether
the last point is above or below this hyperplane. The order type is defined
also for point sets not in general position. In this case the last point can be
below, above or on the corresponding hyperplane.

Suppose that for each $A[i]$ we are given the number of elements of $A$ that
are to the left of $A[i]$, from this information alone we can sort $A$
and recover the left-right relationships mentioned above. Remarkably, this also holds
for higher dimensions. For every pair of indices $i$ and $j$, let $\lambda(i,j)$ be
the number of elements of $S$ above the directed line from $i$ to $j$. The \emph{$\lambda$-matrix}
of $S$ is the $n\times n$ matrix whose $(i,j)$ entry is equal to $\lambda(i,j)$. Goodman
and Pollack~\cite{multidimSorting} showed that from the $\lambda$-matrix of a point
set one can recover the above-below relationships of its triples. This also holds
in dimension $d>2$: if for a given $n$-point set in $\mathbb{R}^d$, one is given the number of points above the oriented
hyperplane defined by every $d$-tuple of points, one can recover which
points are above which oriented hyperplanes.

The $\lambda$-matrix of  set of $n$ points in the plane can be codified with
$O(n^2 \log n)$ bits. This implies that if $f(n)$ is the number
of different possible order types of a set of $n$ points in general position in the plane
then $f(n) \le \exp(O(n^2 \log n))$. Goodman and Pollack~\cite{upper}
showed that this bound is far away from the real value of $f(n)$. They showed that 
 \[f(n) \le \exp(4(1+O(1/\log n))n \log n).\]
To lower bound $f(n)$, consider the following procedure (see~\cite{upper}). Suppose that 
we want to extend $S$ to an $(n+1)$ point configuration
by adding a point $p_{n+1}$ to $S$.
Consider the line arrangement spanned by all the straight
lines passing through a pair of points in $S$.
It was proved by Zaslavsky~\cite{number_of_cells}  that this line arrangement has 
\[\binom{\binom{n}{2}}{2}+\binom{\binom{n}{2}}{1}+1-n\binom{n-2}{2}\ge \frac{1}{8}n^4,\]
cells. Adding $p_{n+1}$ in different cells of the arrangement produces point sets with
different order types. We may use this argument by starting from $\{p_1, p_2, p_3\}$ 
and iteratively adding the remaining points; at each step we consider the number of different options
that produce different order types. This yields 
\[f(n)\ge \prod_{i=1}^n \frac{1}{8}i^4 = \frac{n!^4}{8^n}=\exp(4(1+O(1/\log n))n \log n),\]
were the last term is obtained by using Stirling's formula. 

The order type of a point configuration abstracts the convexity relationships between
its subsets. 
As a result, for various questions regarding point sets, two point sets having the same
order type are equivalent. However, an arbitrary assignment of ``above'' or ``below''
relationships to triples of indices in $\{1,\dots,n\}$ might not be realizable as the order type
of a labeled set of $n$ points. 

Aichholzer, Aurenhammer and Krasser~\cite{database} have produced a database with a point set 
for each realizable order type of at most $10$ points. Although it is a relatively small value
of $n$, this database has proven to be very useful.
Chazelle asked in 1987 (see~\cite{exponential}): what is the number of bits needed to store a representative
of any given realizable order type of $n$ points? Equivalently, what is the minimum size of an integer grid, so that
it contains a representative of every realizable order type of $n$ points?
Goodman, Pollack and Sturmfels~\cite{exponential,intrinsicSpread} showed that 
there are order types of $n$ points whose every realization with positive integer coordinates has a coordinate
of size greater than  $2^{2^{c_{1}n}}$, for some positive constant $c_1$; they also showed that every order type of $n$ points 
can be realized with positive integer coordinates of size at most $2^{2^{c_{2}n}}$, for some positive
constant $c_2$. 
In the book ``Research Problems in Discrete Geometry''~\cite{research_problems}
by Brass, Moser and Pach, we find the following problem.
\begin{problem}\label{prb:nk}
 For a given constant $\alpha > 0$, what is the number of order types of $n$ points that can be
 represented by integer coordinates smaller than $n^{\alpha}$?
\end{problem}

In this paper we show the first non trivial lower bounds for Problem~\ref{prb:nk}. Let 
$g(n,\alpha)$ be the number of different order types realizable in an integer grid
of size $n^{\alpha}$.

For starters one may ask what is the smallest integer grid in which at least
one order type is realizable. This is equivalent to ask what is the size
of the minimum integer grid so that it contains a set of $n$ points
such that no three of them are collinear. This is known as the \emph{no-three-in-a-line problem}
and was introduced by Dudeny~\cite{no_three} in 1917.

\begin{figure}
\begin{center}
\includegraphics[width=0.5\textwidth]{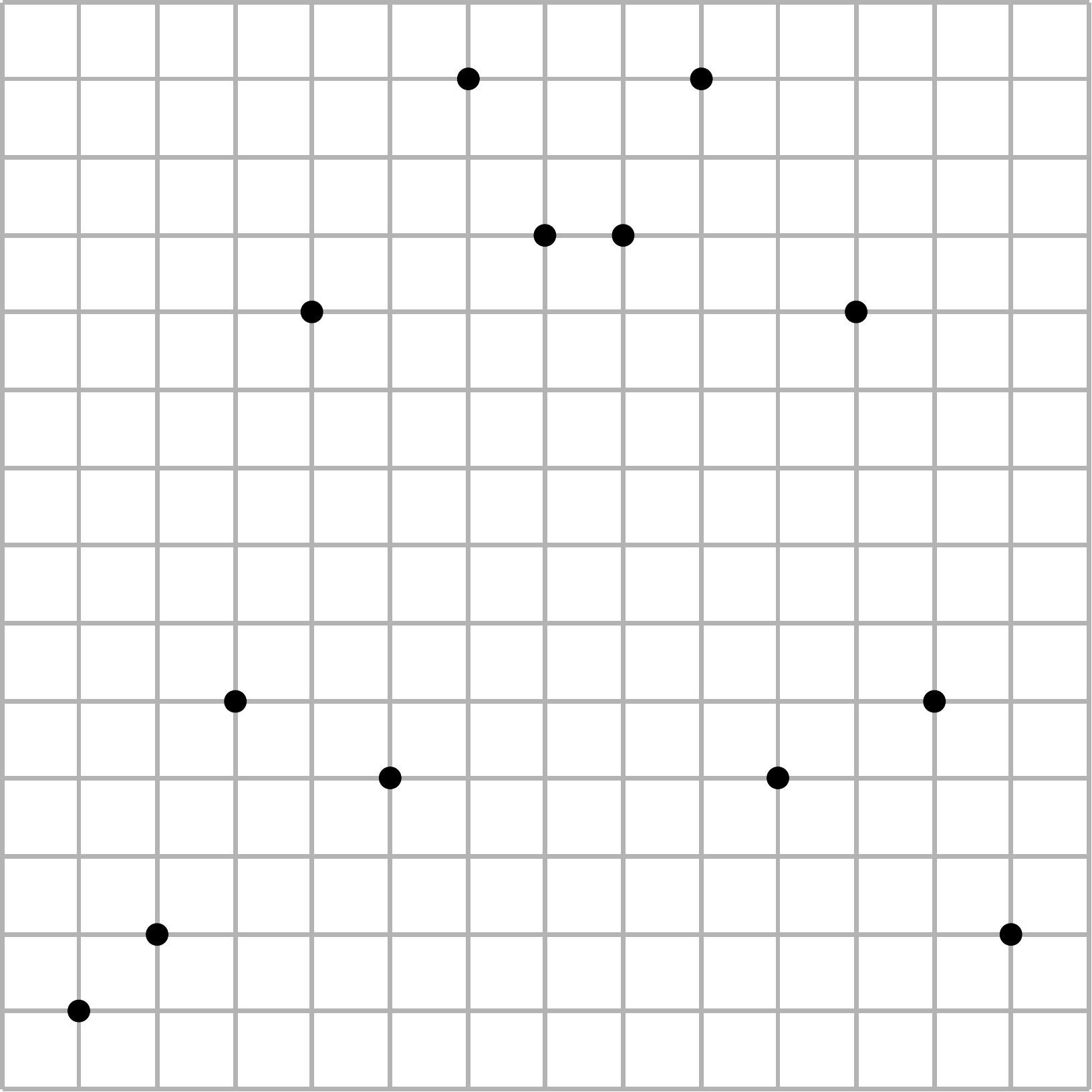}
\end{center}
\caption{$Q_{13}$}\label{fig:Q}
\end{figure}

Erd\H{o}s 
showed (see~\cite{roth}) that if $p$ is a prime then the set \[Q_p:=\{(i, i^2 \mod p): 0 \le i < p \}\]
is in general position. This point set is shown in Figure~\ref{fig:Q} for $p=13$. Therefore, at least one order type can be realized
in integer grids of linear size.

Suppose that $n^{\alpha}$ is such that at least one order type of $n$ points
can be realized in an $n^{\alpha} \times n^{\alpha}$ integer grid. Any permutation of the labels of a point set that
preserves the order type must preserve the clockwise cyclic order of the points in the convex hull. 
Moreover, for every point $p \in S$, the clockwise cyclic order by angle
of the points of $S \setminus \{p\}$ around $p$, must be also be preserved. These two observations together
imply that at least $(n-1)!$ other
different order types are realizable in this grid. By Stirling's approximation,  this at least
\[\exp(n \log n-n+O(\log n)).\] As a result we consider a meaningful lower bound for $g(n,\alpha)$
to be of \[\exp(c\cdot n \log n)\]
for some $c>1$. In this paper, in Section~\ref{sec:cons}, we prove the following lower bounds.

\begin{theorem}\label{thm:a2}
 If $\alpha >2$  then \[g(n,\alpha) \ge \exp\left( 2n \log n-O(n\log \log n)\right). \]
\end{theorem}

\begin{theorem}\label{thm:a2.5}
 If $\alpha \ge 2.5$ then \[g(n,\alpha) \ge \exp(3n \log n-O(n \log \log n)). \]
\end{theorem}

We have the following upper bounds.
Note that there are at most $n^{2n}=\exp(2n\log n)$
different sets of $n$ points in an $n \times n$ integer grid. 
Thus \[g(n,1) \le \exp(2n \log n).\]
By using the point sets found in~\cite{exponential}, one can produce many point sets whose order types 
cannot be realized in an integer grid of size $n^{\alpha}.$ Let $P$ be a point set of  $\log \left (  \alpha \log n \right )$
points whose order type cannot be realized with integer coordinates smaller than $n^\alpha$.
Consider $P \cup Q$,  where $Q$ is any point set of $n- \log \left (  \alpha \log n \right )$
points such that $P \cup Q$ is in general position; note that $P \cup Q$ cannot be realized with integer coordinates smaller than $n^\alpha$.
Therefore, for every $\alpha >0$, there are at least 
\[f\left (n-\frac{\log \left ( \alpha \log n \right )}{c_1} \right)\] 
realizable order types of $n$ points but not realizable in integer grids of size $n^\alpha$.

\section{Lower Bound Constructions}\label{sec:cons}

In this section we prove Theorems~\ref{thm:a2} and~\ref{thm:a2.5}; we present two constructions
that produce many point sets with different order types 
in integer grids of size $n^{\alpha}$ for $\alpha >2$ and $\alpha \ge 2.5$, respectively. Our approach is similar to the one used 
to lower bound $f(n)$: we iteratively place points and lower bound the number of different
available choices that produce different
order types. With the caveat that if we now consider the line arrangement spanned by the straight
lines passing through pairs of already placed points, a given cell might not contain a grid point. 

To work around this problem, we do the following.
We place a portion of our points in a special configuration $\mathcal{C}$; and choose a set of straight lines
passing through pairs of points in $\mathcal{C}$. Then, we define a set $T$ of isothethic squares of side length  equal to $\ell$ 
such that any two squares are separated by one of our chosen lines. Afterwards, we  place the remaining points. 
This is done as follows. At each step we first choose a square  from $T$ that
\begin{itemize}
 \item[$(1)$]  has not been chosen before; and
 \item[$(2)$]  contains a point $p$ of integer coordinates that does not produce
 a triple of collinear points with the previously placed points.
\end{itemize}
We then choose $p$ as our  next point. 

Our strategy is to lower bound, at each step, the number of squares in $T$ that satisfy $(1)$ and $(2)$.
We say that these squares are \emph{alive}; otherwise, we say that they are \emph{dead}.
Suppose that a square of $T$
that has not been chosen yet. If less than $\ell$
lines passing through a pair of previously placed points intersect this square, then
it is still alive. In what follows, we use this observation extensively.

\subsection{Cross Configuration}

Let $n$ be an arbitrarily large positive integer and let $p$ be the smallest prime greater than $n/ 4\log n$. 
 In this case the configuration
$\mathcal{C}$ consists of four sets $\mathcal{U}$, $\mathcal{L}$, $\mathcal{R}$ and $\mathcal{D}$; each set is an affine copy of $Q_p$. 
$\mathcal{L}$ and $\mathcal{R}$ are rotated by $\ang{90}$ and stretched vertically. $\mathcal{U}$ and $\mathcal{D}$ are stretched horizontally.
$\mathcal{L}$ and $\mathcal{R}$ are placed at the same height, with $\mathcal{L}$ to the left of $\mathcal{R}$;
$\mathcal{U}$ and $\mathcal{D}$ are placed
at the same $x$-coordinate and between $\mathcal{L}$ and $\mathcal{R}$; $\mathcal{U}$ is above $\mathcal{L}\cup \mathcal{R} \cup \mathcal{D}$
and $\mathcal{D}$ is below $\mathcal{L}\cup \mathcal{R} \cup \mathcal{U}$. 
Every point in $\mathcal{U}$ is joined with a straight line
with the point in $\mathcal{D}$ with the same $x$-coordinate; 
every point in $\mathcal{L}$ is joined with a straight line
with the point in $\mathcal{R}$ with the same $y$-coordinate. These
are our chosen set of lines. See Figure~\ref{fig:grid}.
Let $k:=\lceil \log n \rceil$. The precise definitions are 
\begin{align}
  \mathcal{U} & := \{(i\cdot 34p\cdot k^2, (34\cdot (i^2 \mod p))\cdot k^2): 0 \le i < p\}, \nonumber \\
  \mathcal{L} & := \{((34\cdot (i^2 \mod p)-136p^2) \cdot k^2, (i\cdot 68p-238p^2)\cdot k^2): 0 \le i < p\}, \nonumber \\
  \mathcal{R} & :=  \{((34\cdot (i^2 \mod p)+153p^2) \cdot k^2, (i\cdot 68p-238p^2)\cdot k^2): 0 \le i < p\}, \nonumber \\
  \mathcal{D} & := \{(i\cdot 34p\cdot k^2, (34\cdot (i^2 \mod p)-408p^2)\cdot k^2): 1 \le 0 < p\} \textrm{ and } \nonumber \\
  \mathcal{C} & :=\mathcal{U} \cup \mathcal{L} \cup \mathcal{R} \cup \mathcal{D}. \nonumber
\end{align}

\begin{figure}
\begin{center}
\includegraphics[width=0.5\textwidth]{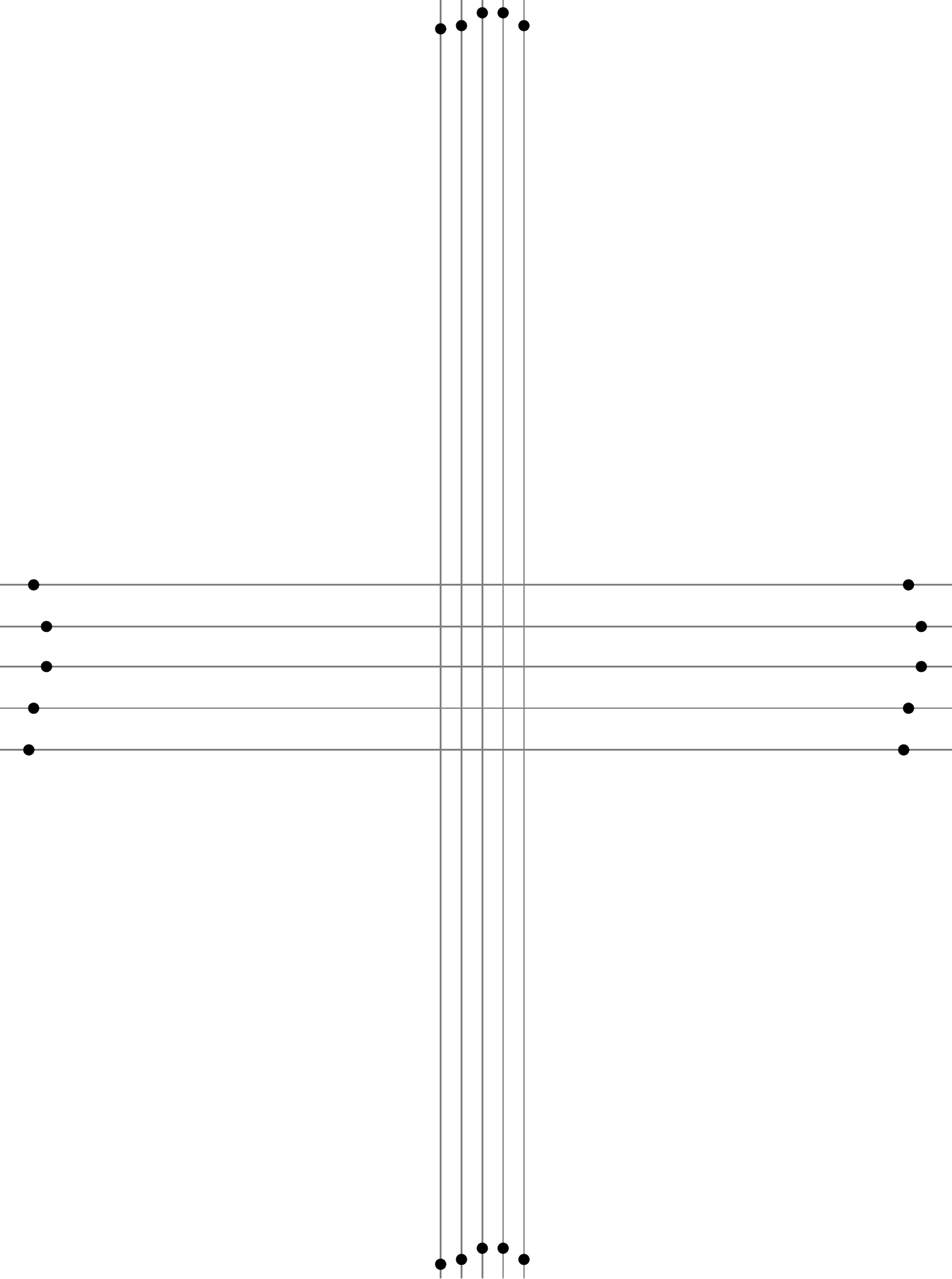}
\end{center}
\caption{Cross configuration for $p=5$}\label{fig:grid}
\label{fig:T}
\end{figure}

Simple (but tedious) arithmetic shows that $\mathcal{C}$ is in general position.
The set of chosen straight lines form a rectangular grid.  In the interior of each of these rectangles place an isothethic square 
with $32pk^2 \times 32pk^2$ integer grid points.
Let $T$ be the set of these squares. 
Baker, Harman and Pintz \cite{primegaps} showed that the interval $[x, x+x^{21/40}]$ contains a prime number, for $x$ sufficiently large. Thus,
$p = n/4 \log n + O(n^{21/40})$. Therefore $|\mathcal{C}|=n/\log n+ O(n^{21/40})$, $|T|=(p-1)^2$ and $\ell=32pk^2$ for this construction.

We now iteratively place the remaining $n-4p$ points. At each stage the number of lines passing through a pair of the so far placed
points is less than $n^2/2$; each of these lines intersects
less than $2p$ squares of $T$; each square must touched by at least
$32pk^2$ straight lines before being dead. Therefore, the number of alive squares at every stage
is at least
\[(p-1)^2-\frac{n^2p}{32pk^2}=\frac{1}{2}p^2-O(p) \ge \frac{n^2}{32 \log^2 n}.\]
where the last inequality holds for sufficiently large $n$.

Therefore, we obtain at least 
\[\prod_{i=1}^{n-4p} \frac{n^2}{32 \log^2 n}=\frac{n^{2(n-4p)}}{(32 \log^2 n)^{n-4p}}=\exp\left( 2n \log n-O(n\log \log n)\right)\]
different order types with this procedure. 
Since $\mathcal{C}$ is contained in an integer grid of side length equal to $\Theta(p^2 k^2)=\Theta(n^2)$, this proves Theorem~\ref{thm:a2}.

\subsection{Regular Polygon Configuration}
Let $n$ be an arbitrarily large positive integer; let $m$ be the smallest multiple of $16$ larger than $n/\log n$ and let $L:=\lceil 64 n^2/m^{1/2}\rceil$. 
Let $\mathcal{C}:=\{v_0,\dots, v_{m-1}\}$ be the vertices, in clockwise order,
of a regular polygon $P$ of side length equal to $L$. These points may not have integer coordinates; their coordinates
will be rounded up to the nearest integer later on.
Let $q^*$ be the center of this polygon. For $1 \le i \le m$,  let $\triangle_i$ be the triangle with vertices
$v_{i-1}, v_{i},$ and  $q^*$. In what follows we define a set $T_i$ of squares 
inside $\triangle_i$.

\begin{figure}[h]
	\begin{center}
		\includegraphics[width=0.75\textwidth]{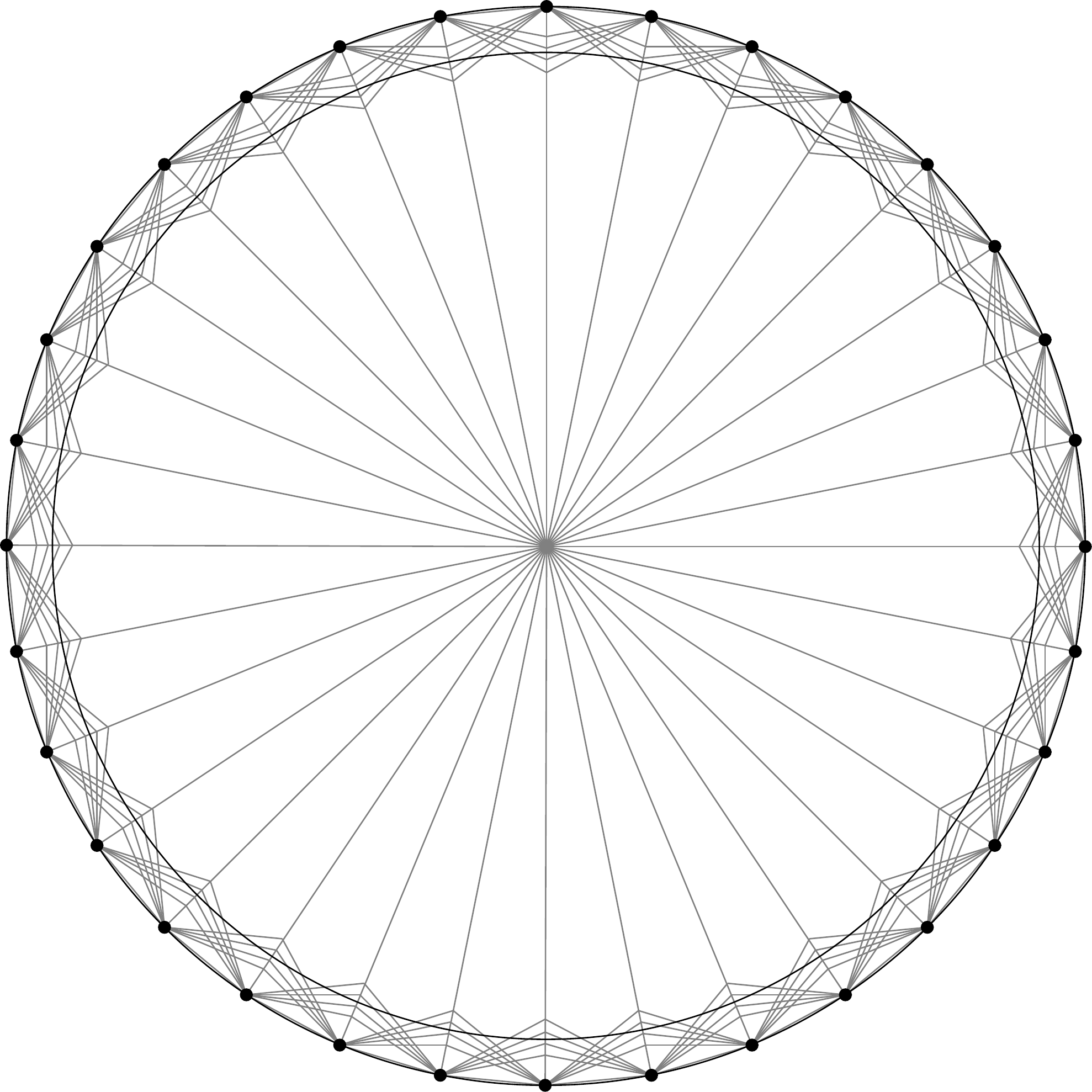}
	\end{center}
	\caption{The regular polygon construction with $m=32$}\label{fig:construction}
\end{figure}

Starting at the line segment joining $v_{i-1}$ and $v_{i}$, 
let $e_1,\dots,e_{m-1}$ be the line segments joining $v_{i-1}$ and every other vertex of $P$, sorted clockwise by angle around $v_{i-1}$.
Starting at the line segment joining $v_{i}$ and $v_{i-1}$, 
let $f_1,\dots,f_{m-1}$ be the line segments joining $v_i$ and every other vertex of $P$, sorted counterclockwise by angle around $v_i$.
Let $C_1$ be the the circumcircle of $P$. Since
every pair of consecutive vertices of $P$ defines a chord of
$C_1$, and these chords have the same length, the angle between any two consecutive $e_j$ and $e_{j+1}$
is the same. Let $\gamma$ be this angle. Moreover, the angle between any two consecutive $f_j$ and $f_{j+1}$ is also equal to $\gamma$; note that 
 \[\gamma=\frac{1}{m}\pi.\] For  indices $2 \le j \le m/2$ and $2 \le k \le m/2$, let $p_{j,k}$ be the intersection of $e_j$ and $f_k$; 
 note that $p_{j,k}$ is contained in $\Delta_i$. 
 Let \[Q:=\left \{p_{j,k}: j,k \textrm{ even and } \frac{m}{8} \le j,k < \frac{m}{4} \right \}.\] Note
 that \[|Q| = \frac{m^2}{256}.\]See Figure~\ref{fig:construction}. 
 For each $p_{j,k}$ in $Q$, place an isothethic square
of side length equal to 
\[\ell:=\frac{L}{m}\] 
centered at $p_{j,k}$.  Let $T_i$ be the set of these squares. 
The next lemma shows that the squares in $T_i$ are well separated by the $e_j$'s and $f_k$'s
%

\begin{figure}[ht]
	\begin{center}
		\includegraphics[page=1, width=.5\textwidth]{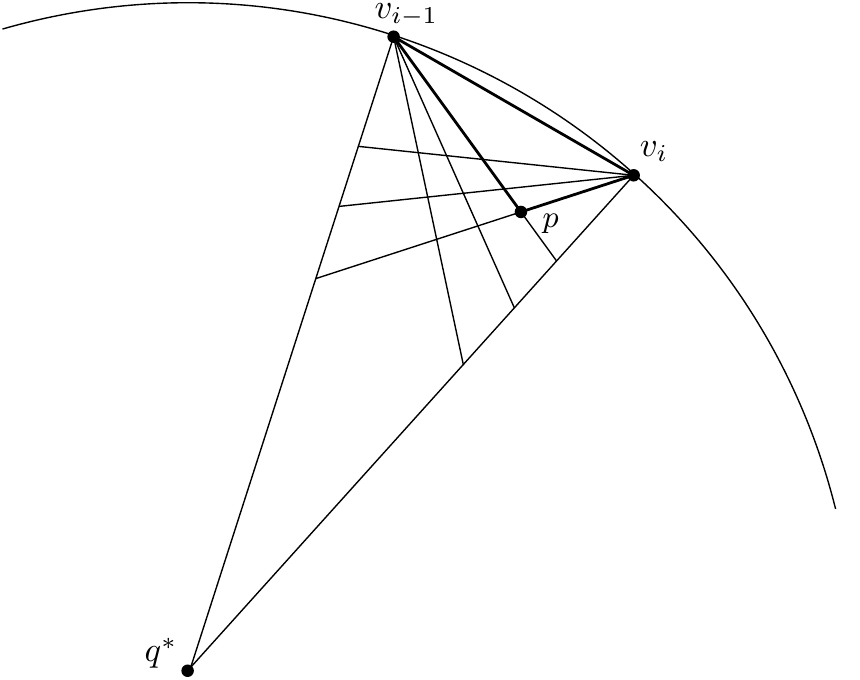}
	\end{center}
	\caption{The proof of Lemma \ref{lem:separated}}\label{fig:zoom1}
\end{figure}

\begin{lemma}\label{lem:separated}
Let $p_{j,k}$ be a point in $Q$. Then the distances from $p_{j,k}$ to $e_{j-1}$, $e_{j+1}$, $f_{k-1}$ and $f_{k+1}$
 are greater than
 \[\left (\sqrt{2}-1 \right )\pi \ell+O \left ( \frac{\ell}{m^2} \right ).\]
\end{lemma}
\begin{proof}
 We show that the distances from $p_{j,k}$ to $e_{j-1}$ and $e_{j+1}$ are at least the required value.
 The proof for $f_{k-1}$ and $f_{k+1}$ is similar. Note that among the $p_{j,k}$'s in $Q$, 
 \[p:=p_{m/8,  m/4-1}\] is the point
 closest to $v_{i}$. Consider the triangle with vertices $v_{i}, v_{i-1}$ and $p$ (see Figure \ref{fig:zoom1}). By the law
 of sines the distance from $p$ to $v_{i}$ is equal to 
 \[\frac{\sin(\pi/8)}{\sin(5\pi/8+\pi/m)}L > \frac{\sin(\pi/8)}{\sin(5\pi/8)}L= \left (\sqrt{2}-1 \right )L.\]
 Therefore, the distances from $p_{j,k}$ to $e_{j-1}$ and $e_{j+1}$ are at least  
 $\tan(\gamma) \cdot (\sqrt{2}-1)L$. The result follows from the facts that Maclaurin series
 of $\tan(x)$ is equal to $x+O(x^3)$ and that $\gamma=\pi/m$.
\end{proof}

We are ready to define the set of squares, let \[T:=\bigcup_{i=1}^m T_i.\] Let $C_2$ be the circle with center $q^*$ and passing through 
$p_{m/4, m/4}$. Note that $T$ is contained in the 
annulus $A$ bounded by $C_1$ and $C_2$.
The following lemma upper bounds the number of squares in $T$ that a given straight line can intersect.

\begin{figure}[ht]
	\begin{center}
		\includegraphics[page=2, width=.5\textwidth]{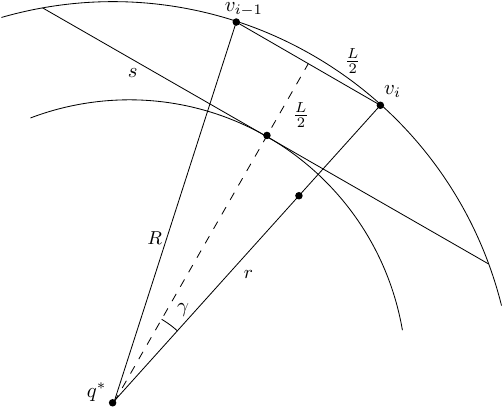}
	\end{center}
	\caption{The proof of Lemma \ref{lem:squaresIntersected}}\label{fig:zoom2}
\end{figure}

\begin{lemma}\label{lem:squaresIntersected}
 Every straight line intersects at most
 \[\frac{m^{3/2}}{4}\]
 squares of $T$. 
\end{lemma}
\begin{proof}
Let $\varphi$ be a straight line. Note that $\varphi$ intersects $A$ in at most two straight line segments. We upper bound
the number of squares in $T$ that a straight line segment $s$ can intersect.
Each time $s$ intersects a square in  $T_i$, it must intersect and edge
$e_j$ or $f_k$; moreover, only half of these edges define a square in $T_i$. Therefore,  $s$ intersects at most $\frac{m}{8}$ squares in  $T_i$. 
We upper bound the number of the triangles $\triangle_i$'s that $s$ can
intersect. For this we upper bound the length of $s$.

Let $R$ and $r$ be the radius of the circles $C_1$ and $C_2$, respectively.
Note that $s$ has maximum length when it is tangent to $C_2$ and its endpoints are in $C_1$.
Therefore,
\[||s||\le 2 \sqrt{R^2-r^2}.\]
Since $C_2$ passes
through $p_{m/4, m/4}$, the distance from $C_2$ to the edge $v_i,v_{i-1}$
is equal to $L/2$. Since $P$ is a regular polygon $R=\frac{1}{2}L \csc(\pi/m)$ and its apotheme 
is equal to $R \cos (\pi /m)$ (see Figure \ref{fig:zoom2}). This implies that $r = \frac{1}{2}L(\cot(\pi/m)-1)$.
Therefore,  \[||s||\le 2 \sqrt{R^2-r^2} \le  \sqrt{2}\cdot L\sqrt{\cot \left  (\frac{\pi}{m}\right ) } \le \sqrt{\frac{2m}{\pi}}L-O\left (\frac{L}{m^{3/2}}\right );\]
the last term comes from the fact that the Maclaurin series of $\sqrt{\cot(x)}$ is equal to $\sqrt{\frac{1}{x}}-O(x^{3/2})$.

Now we lower bound the length of $s\cap \triangle_i$.
Note that $s\cap \triangle_i$ has minimum length when $s$ is tangent to $C_2$ and parallel to the edge $v_{i-1},v_{i}$.
Thus, 
\[||s\cap \triangle_i||\ge 2 \tan \left (\frac{\pi}{m} \right )r=\left ( 1-\tan \left (\frac{\pi}{m} \right)\right)L > \sqrt{\frac{2}{\pi}}L,\]
where the last term holds for sufficiently large $n$.
Therefore, $s$ intersects a most 
$\sqrt{m}$ of the triangles $\triangle_i$. The result follows.
\end{proof}

To end the construction we round the coordinates of the $v_i$'s to their nearest integer. Redefine
the $e_j$'s and $f_k$'s accordingly. By Lemma~\ref{lem:separated}, a square in $T_i$ centered at $p_{j,k}$
is separated from  edges $e_{j'}$ and $f_{k'}$ different $e_j$ and $f_k$ by a distance of at least
$\left (\sqrt{2}-1 \right ) \pi \ell$. The endpoints of the new $e_j$'s and $f_k$'s are 
at a distance of at most one of their original positions. Since $\left (\sqrt{2}-1 \right )\pi > 1$,
the squares in $T_i$ are still separated by the straight lines containing the $e_j$'s and $f_k$'s.

We now iteratively place  the remaining $n-m$ points. At every stage the number of lines passing
through every pair of the so far placed points is less than $n^2/2$; each of these lines intersects 
at most $m^{3/2}/4$ squares of $T$; each square must be touched by at least $\ell=L/m$ straight lines before
being dead. Thus, the number of squares alive at every stage is at least
\[\frac{m^3}{256}-\frac{n^2m^{5/2}}{8L} \ge \frac{m^3}{512} \ge \frac{n^3}{512 \log^3 n}.\]
Therefore, we obtain at least
\[\prod_{i=1}^{n-m} \frac{n^3}{512 \log^3 n}=\frac{n^{3(n-m)}}{(512 \log n)^{3(n-m)}}=\exp(3n \log n-O(n \log \log n))\]
different order types with this procedure. 
Recall that $m \le n / \log n+16$ and $L \le 64 n^2/m^{1/2}+1$.
Therefore, these point sets lie in an
integer grid of side length equal to $L\cdot m =\Theta \left (n^{2.5}/\sqrt{\log n} \right )$. This proves Theorem~\ref{thm:a2.5}.

\textbf{Acknowledgments.} 
This work was initiated at the \emph{Third Workshop on Geometry and Graphs}, held at the Bellairs
Research Institute, Barbados, 2015. We are grateful to the other workshop participants
for providing a stimulating research environment; in particular we thank Jean-Lou De Carufel and Stefanie Wuhrer for various helpful discussions.
This work was continued at the \emph{Winter School on Computational and Combinatorial Geometry at University of Havana},  Cuba, 2016.

\small
\bibliographystyle{abbrv} \bibliography{small}



\end{document}